\documentclass[secthm,seceqn,amsthm,ussrhead]{amsart}
\usepackage{amsmath,latexsym}
\usepackage[psamsfonts]{amssymb}
\usepackage{times}
\usepackage[mathcal]{euscript}
\numberwithin{equation}{section} \textwidth=135mm \textheight=200mm
\newcommand{\bbC}{\mathbb C}

\renewcommand{\epsilon}{\varepsilon}

\newcommand{\be}{\begin{equation}}
\newcommand{\ee}{\end{equation}}

\newcommand{\R}{\mathbb{R}}

\newcommand{\T}{\mathbb{T}}

\newcommand{\Z}{\mathbb{Z}}

{\bf}{\it}
\newtheorem{theorem}{Theorem}[section]
\newtheorem{lemma}[theorem]{Lemma}
\newtheorem{corollary}[theorem]{Corollary}
\newtheorem{hypothesis}[theorem]{Hypothesis}
\newtheorem{definition}[theorem]{Definition}
\newtheorem{proposition}[theorem]{Proposition}
\newtheorem{remark}[theorem]{Remark}

\date{\today}
\begin{document}

\title[Saidakhmat N.Lakaev,\,\, Maslina Darus,\,\, Said T. Dustov]{Threshold phenomenon
for a family of the Generalized Generalized Friedrichs models with
the perturbation of rank one}
\author{Saidakhmat N.\,Lakaev,\,\, Maslina Darus,\,\,Said T. \,Dustov}
\maketitle
\begin{abstract}
A family $H_\mu(p),$ $\mu>0,$ $p\in\T^3$  of the Generalized
Firedrichs models with the perturbation of rank one, associated to a
system of two particles, moving on the three dimensional lattice
$\mathbb{\Z}^3,$ is considered. The existence or absence of the
unique eigenvalue of the operator $H_\mu(p)$ lying outside the
essential spectrum, depending on the values of $\mu>0$ and $p\in
U_{\delta}(p_{\,0})\subset\T^3$ is proven. Moreover, the analyticity
of associated eigenfunction is shown.
\end{abstract}

\section*{Introduction}

In celebrated work  \cite{MK-BS} of B.Simon and M.Klaus it is
considered a family of Schr\"{o}dinger ope\-ra\-tors $H=-\Delta+\mu
V$ and, a situation where as $\mu$ tends to $\mu_0$ some eigenvalue
$e_i(\mu)$ tends to $0,$ i.e., as $\mu$ tends to $\mu_0$ an
eigenvalue is absorbed into continuous spectrum, and conversely, as
$\mu$ tends to $\mu_0+\epsilon, \epsilon>0 $ continuous spectrum
{\it gives birth} to a new eigenvalue. This phenomenon in
\cite{MK-BS} is called {\it coupling constant threshold.}

In \cite{ALMM} for a wide class of  two-body energy operators
$H_2(k)$ on the $d$-dimensional lattice $\Z^d$, $d\ge 3$, $k$
 being the two-particle quasi-momentum, it is proven that if the following two
assumptions (i) and (ii) are satisfied, then for all nontrivial
values  $k$, $k\ne 0$, the discrete spectrum of  $h(k)$ below its
threshold  is non-empty. The assumptions are: (i) the two-particle
Hamiltonian $H_2(k)$ associated to the zero value of the
quasi-momentum has either an eigenvalue or a virtual level at the
bottom of its essential spectrum and (ii) the one-particle free
Hamiltonians in the coordinate representation generate  positivity
preserving semi-groups.

In \cite{LH} the Hamiltonian of a system of two identical quantum
mechanical particles (bosons) moving on the $d$-dimensional lattice
$\Z^d,d\geq3$ and interacting via zero-range repulsive pair
potentials is considered. For the associated two-particle
Schr\"{o}dinger operator $H_{\mu}(K),$ $K\in \T^d=(-\pi,\pi]^d$
there existence of coupling constant threshold $\mu=\mu_0(K)>0$ is
proven:the operator has non eigenvalue for any $0<\mu<\mu_0,$ but
for each  $\mu>\mu_0$ it has  a unique eigenvalue $z(\mu,K)$  above
the upper edge of the essential spectrum of $H_{\mu}(K).$ Moreover
asymptotics for $z(\mu, K)$ are found, when $\mu$ approaches to
$\mu_0(K)$  and $K\to 0.$

Notice that in \cite{MK-BS} the existence of a coupling constant
threshold has been assumed, at the same time in \cite{LH} the
coupling constant threshold is definitely found by the given data of
the Hamiltonian.

Notice also that for the  Hamiltonians of a system of two identical
particles moving on $\R^2$ or $\Z^2$  the coupling constant
threshold vanishes, if particles are bosons and the coupling
constant threshold is positive, if particles are fermions.

In the present paper, a family of Generalized Friedrichs models
under rank one perturbations $H_{\mu}(p)$, $ \mu>0, p\in
U_{\delta}(p_{\,0})\subset \T^3,$ where $U_{\delta}(p_{\,0})$ is a
$\delta$-neighborhood of the point $p=p_{\,0}\in\mathbb{T}^3,$
associated to a system of two particles on the three-dimensional
lattice $\Z^3$ interacting via pair local {\it repulsive potentials}
is considered.

If parameters of the Generalized Friedrichs model satisfy some
conditions then there exists a {\it coupling constant threshold}
$\mu=\mu(p)>0$ that the operator has non eigenvalue for any
$0<\mu<\mu(p),$ but for any $\mu>\mu(p)$ there is a unique
eigenvalue $z(\mu,p)$ of $H_{\mu}(p)$, which lie above the threshold
$z=M(p)$ of the operator $H_\mu(p),$  $p\in U_{\delta}(p_{\,0}).$
For the associated eigenfunction an explicit expression is found and
its analyticity is proven.

We have found necessary and sufficient conditions, in order to the
threshold $z=M(p)$ was an eigenvalue or a resonance (virtual level)
or a regular point of the essential spectrum of $H_\mu(p),$ $p\in
U_{\delta}(p_{\,0}).$

One of the reasons to consider the family of the Generalized
Friedrichs models interacting via pair local {\it repulsive
potentials} is as follows: the family of the Generalized Friedrichs
models generalizes and involves some important behaviors as of the
Shr\"odinger operators associated to the Hamiltonians for systems of
two arbitrary particles moving on $\R^d$ or $\Z^d,d\geq 1,$ as well
as, the Hamiltonians for systems of both bosons and fermions
\cite{LAQ12},\,\cite{Lak83, Lak86},\,\cite{LMfa04}.

Furthermore, as have been stated in \cite{DenschlagDaley,RBAP} that
throughout physics stable composite objects are usually formed by
way of attractive forces, which allow the constituents to lower
their energy by binding together. The repulsive forces separates
particles in free space. However, in structured environment such as
a periodic potential and in the absence of dissipation, stable
composite objects can exist even for repulsive interactions.

The family of the Generalized Friedrichs models theoretically
adequately describes this phenomenon relating  to {\it repulsive
forces,} since the two-particle discrete Schr\"odinger operators are
the special case of this family.

The Generalized Friedrichs model, i.e., the case, where the
non-perturbed operator $H_0$ is a multiplication operator by
arbitrary function with Van Hove singularities(critical points)
defined on the closed interval $[a,b]$ has been considered in
\cite{Lakaev86}. In this case, the multiplicity of continuous
spectrum is not constant.

Generalized Friedrichs model with given number of eigenvalues
embedded in the continuous spectrum has been constructed
\cite{AIL95}.

The Generalized Friedrichs models appear mostly in the problems of
solid state physics \cite{Mogilner, RSIII}, quantum mechanics
\cite{Fri48},
 and quantum field theory \cite{Fri08, MalMin} and in
general settings have been studied in \cite{Lak83, Lak86}.

In \cite{ALzMJMAA} the family of Generalized Friedrichs models under
rank one perturbations $H_{\mu}(p),\,\mu>0,\,p\in(-p,p]^3,$
associated to a system of two particles on the three-dimensional
lattice $\Z^3$ is considered. In some special case of multiplication
operator and under the assumption that the operator
$H_{\mu}(0),0\in\T^3$ has a {\it coupling constant threshold
$\mu_0(0)>0$} the existence of a unique eigenvalue below the bottom
of the essential spectrum of $H_{\mu_0(0)}(p),p\in(-p,p]^3$ for all
non-trivial values of $p\in\T^3$  has been proved.

In \cite{LAQ12} for a family of the Generalized Friedrichs models
$H_{\mu}(p)$, $ \mu>0, p\in\T^2$ either the existence or absence of
a positive coupling constant threshold $\mu=\mu(p)>0$ depending on
the parameters of the model has been proved.

In \cite{LAQ} it is established  an expansion of the threshold
eigenvalue $E(\mu,p)$ and resonance in some neighborhood of the
point $\mu=\mu(p)$ .

In \cite{LMfa04} a special family of the Generalized Friedrichs
models has been  considered and the existence of eigenvalues for
some values of quasi momentum $p\in\mathbb{T}^d$ of the system,
lying in a neighborhood of some $p_{\,0}\in\mathbb{T}^d,$ has been
proved.

\section{Preliminary notions and  assumptions. Formulation of the main results.}

Let $\mathbb{Z}^3$ be the three-dimensional hybercubic lattice and
$$\mathbb{T}^3=(\mathbb{R}/2\pi\mathbb{Z})^3=(-\pi,\pi]^3$$ be the three-dimensional
 torus (Brillion zone), the dual group of $\mathbb{Z}^3.$

Note that operations addition and multiplication by number of the
elements of torus ${\mathbb{T}^3} \equiv(-\pi,\pi]^3
\subset\mathbb{R}^3 $ is defined as operations in $\R^3$ by the
module $(2\pi \mathbb{Z}^3).$

Let $L^2(\mathbb{T}^3)$ be the Hilbert space of square-integrable
functions defined on the torus $\mathbb{T}^3$ and  $\mathbb{C}^1$ be
one-dimensional complex Hilbert space.

We consider a family of the Generalized Friedrichs models acting in
$L^2(\T^3)$ as follows:
$$H_{\mu}(p)=H_0(p)+\mu\Phi^*\Phi,\,\,\mu>0.$$
Here $$\Phi:L^2(\T^3)\to \mathbb{C}^1,\quad \Phi
f=(f,\varphi)_{L^2(\T^3)},$$ $$\Phi^*: \mathbb{C}^1\to
L^2(\T^3),\quad \Phi^* f_0=\varphi(q) f_0,$$ where
$(\cdot,\cdot)_{L^2(\T^3)}$ -- inner product in $L^2(\T^3)$ and
$H_0(p),$ $p\in\T^3$ is a multiplication operator by a function
$w_p(\cdot):=w(p,\cdot),$ i.e.
\begin{equation}\label{h0}
(H_0(p)f)(q)=w_p(q)f(q),\quad f\in L^2(\T^3).
\end{equation}

Note that for any$f\in L^2(\T^3)$ and $g_0\in\mathbb{C}^1$ the
equality $$(\Phi f,g_0)_{\mathbb{C}^1}=(f,\Phi^*g_0)_{L^2(\T^3)}$$
holds. The following assumption will be needed throughout the paper.
\begin{hypothesis}\label{Hyp1}
We assume that the following assumptions are satisfied:
\begin{itemize}
\item[(i)] the function $\varphi(\cdot)$ is nontrivial, real-analytic
function on $\T^3;$
\item[(ii)] the function $w(\cdot,\cdot)$ is real-analytic function on
$(\T^3)^2=\T^3\times \T^3$ and has a unique non degenerated maximum
at $(p_{\,0}, q_0)\in (\T^3)^2$.
\end{itemize}
\end{hypothesis}

The perturbation $v=\Phi^{\ast}\Phi$ is positive operator of rank 1.
Consequently, by the well-known Weyl theorem \cite{RSIV} the
essential spectrum fills the following segment on the real axis:
$$
\sigma_{ess}(H_\mu(p))=\sigma_{ess}(H_0(p))=[m(p),\,M(p)],
$$
where
\begin{equation*}
m(p)=\min_{q\in \T^3}w_p(q),\quad M(p)= \max_{q\in \T^3}w_p(q).
\end{equation*}

By Hypothesis \ref{Hyp1} there exist such $\delta$-neighborhood
$U_{\delta }(p_{\,0})\subset \T^3$ of the point $p=p_{\,0}\in \T^3$
and analytic vector function $\mathbf{q}_0:U_{\delta}(p_{\,0})\to
\T^3$ that for any $p\in U_{\delta}(p_{\,0})$ the point
$\mathbf{q}_0(p)=(q_0^{(1)}(p), q_0^{(2)}(p), q_0^{(3)}(p))\in\T^3$
is a unique non degenerated maximum of the function $w_p(\cdot)$
(see Lemma \ref{non degenerate}

Moreover, the following integral
$$\frac{1}{\mu(p)}=\int\limits_{\T^3}
\frac{\varphi^2(s)ds}{M(p)-w_p(s)}>0$$ exists (see Lemma
\ref{integral}).

The positive number $\mu(p)>0$  is called {\it coupling constant
threshold}.

\begin{definition}\label{defregular}
The threshold $z=M(p)$  is called a regular point of the essential
spectrum of the operator $H_{\mu}(p),$ if the equation
$H_{\mu}(p)f=M(p)f$ has only trivial solution $f \in L^2(\T^3)$.
\end{definition}
Let $L^1(\T^3)$ be the Banach space of integrable functions on
$\T^3$.
\begin{definition}\label{defvirtual}
The threshold $z=M(p)$ is called a $M(p)$ energy resonance (virtual
level) of the essential spectrum of the operator $H_{\mu}(p),$ if
the equation $H_{\mu}(p)f=M(p)f$ has a non-trivial solution $f\in
L^{1}(\T^3)\setminus L^2(\T^3).$  The solution $f$ is called {\it
resonance state} of the operator $H_{\mu}(p).$
\end{definition}
\begin{remark}
The set $\mathbb{G}$ of $\mu>0$,\, for which the threshold is a
regular point of the essential spectrum $\sigma_{ess}(H_\mu(p))$ of
$H_{\mu}(p),$ is an open set in $(0,+\infty)$. More  precisely,
$\mathbb{G}=(0,+\infty)\backslash \{\mu(p)\}$.
\end{remark}
\begin{remark}\label{regular point}
If the threshold $z=M(p)$ is a regular point of $H_{\mu}(p)$ then
the number of eigenvalues of the operator  $H_{\mu}(p)$ above the
threshold $M(p)$ does not change under small perturbations of
$\mu\in \mathbb G$ (see items $(i)$,\, $(ii)$ and $(iii)$ of Theorem
\ref{main theorem}).
\end{remark}

In the following theorem we have found a necessary  and sufficient
conditions for existence of a unique eigenvalue $E(\mu,p),$ lying
above the threshold of the essential spectrum of  $H_{\mu}(p),$
$p\in U_{\delta}(p_{\,0}).$  We  prove that for a fixed $p\in
U_{\delta}(p_{\,0}),$ the function $E(\cdot,p)$ is analytic in
$(\mu(p),+\infty).$ Moreover for the associated eigenfunction an
explicit expression is found and its analyticity is proven.
Furthermore, in the case $\mu=\mu(p)>0$, it is proven that the
threshold $M(p)$ of the essential spectrum is either a $M(p)$ energy
resonance or eigenvalue for the operator $H_\mu(p),$ $p\in\T^3.$

\begin{theorem}\label{main theorem}
Assume Hypothesis \ref{Hyp1} and  $p\in U_{\delta}(p_{\,0}).$ Then
the following assertions are true.
\begin{enumerate}
\item[{\rm(i)}] The operator
$H_\mu(p)$ has a unique eigenvalue  $E(\mu,p)$ lying above the
threshold $M(p)$ of the essential spectrum if and only if
$\mu>\mu(p)$. The function $E(\cdot,p)$ is monotonously increasing
real-analytic function in the interval $(\mu(p),+\infty )$ and the
function $E(\mu,\cdot)$ is real-analytic in $U_{\delta}(p_{\,0}).$
The associated eigenfunction
$$\Psi(\mu;p,q,E(\mu,p))=\frac{C\mu
\varphi(q)}{E(\mu,p)-w_p(q)}$$ is analytic on $\T^3,$ where $C\neq
0$ is normalization factor. Moreover, the mappings
$$\Psi: U_{\delta}(p_{\,0})\to L^2(\T^3),\quad p \mapsto
\Psi(\mu;p,q,E(\mu,p))\in L^2(\T^3)$$ and
$$\Psi: (\mu(p),+\infty)\to L^2(\T^3),\quad \mu \mapsto
\Psi(\mu;p,q,E(\mu,p))\in L^2(\T^3)$$ are vector-valued analytic
functions in $U_{\delta}(p_{\,0})$ and $(\mu(p),+\infty),$
respectively .

\item[{\rm (ii)}]
The operator $H_\mu(p)$ has none eigenvalue in  semi-infinite
interval $(M(p), \infty)$ if and only if $0<\mu<\mu(p).$

\item[{\rm (iii)}] The threshold $z=M(p)$ is a regular point  of the
operator $H_\mu(p)$ if and only if $\mu\ne\mu(p)$.

\item[{\rm (iv)}] The threshold $z=M(p)$ is a $M(p)$ energy resonance of the
operator $H_\mu(p)$ if and only if $\mu=\mu(p)$ and
$\varphi(\mathbf{q}_0(p))\neq 0.$ The associated resonance state is
of the form
$$f(q)= \frac{C\mu(p) \varphi(q)}{M(p)-w_p(q)},$$ where $C\neq 0$ is a normalizing constant
 and $f\in L^1(\T^3)\setminus L^2(\T^3).$

\item[{\rm (v)}]
The threshold $z=M(p)$ is an eigenvalue of the operator $H_\mu(p)$
if and only if $\mu=\mu(p)$ and $\varphi(\mathbf{q}_0(p))=0.$
Moreover, if the threshold $z=M(p)$ is an eigenvalue of the operator
$H_\mu(p)$ then the associated eigenfunction is of the form
\be\label{kkk} f(q)=\frac{C\mu(p) \varphi(q)}{M(p)-w_p(q)}\in
L^2(\T^3) ,\ee where $C\neq 0$ is a normalizing constant.
\end{enumerate}
\end{theorem}

\begin{remark}\label{positivity}
From the positivity of $\Phi^{\ast}\Phi$ it follows that the
operator $H_{\mu}(p)$ has none eigenvalue lying below $m(p)$.
\end{remark}

\section{Proof of the results}
We postpone the proof of the  theorem after several lemmas and
remarks.
\begin{lemma}\label{non degenerate}
Assume Hypothesis \ref{Hyp1}. Then there exist such
$\delta$-neighborhood  $U_ {\delta } (p_{\,0})\subset \T^3$ of the
point $p =p_{\,0}$ and analytic function $\mathbf{q}_0 :U_ {\delta }
(p_{\,0})\to \T^3$ that for any  $p \in U_ { \delta }(p_{\,0})$ the
point $\mathbf{q}_0(p) $ is a unique non-degenerated maximum of the
function $w_p(\cdot).$
\end{lemma}

\begin{proof}
By Hypothesis \ref{Hyp1}  the square matrix
$$
A(0)=\left(\frac{\partial^{2}w_{p_{\,0}}}{\partial q_{i}\partial
q_{j}}(\mathbf{q}_0)\right)_{i,j=1}^{3}<0
$$
is negatively defined and $\nabla w_{p_{\,0}}(q_0)=0.$ Then by the
implicit function theorem (the analytic case)  there exist a
$\delta$-neighborhood  $U_ {\delta } (p_{\,0})\subset \T^3$ of
$p=p_{\,0}\in \T^3$ and a unique analytic vector function
$\mathbf{q}_0(\cdot):U_{\delta}(p_{\,0})\rightarrow \T^3$ such that
$\nabla w_{p}(q_0(p))=0$ and
$$
A(p)=\left(
 \frac{\partial^{2}w_{p}}{\partial q_{i}\partial
q_{j}}(\mathbf{q}_0(p)) \right)_{i,j=1}^{3}<0,\,\,p\in
U_{\delta}(p_{\,0}).
$$
Hence for any $p\in U_{\delta}(p_{\,0})$ the point
$\mathbf{q}_{0}(p)$ is a unique non degenerated maximum of the
function $w_{p}(\cdot)$.
\end{proof}

For any $\mu>0$ and $p\in \T^3$ we define in $\mathrm{C}\setminus
[m(p); M(p)]$ an analytic function $\Delta(\mu,p; \cdot)$(the
Fredholm determinant $\Delta(\mu,p; \cdot),$ associated to the
operator $H_\mu(p)$) as
\begin{equation}\label{Det.H.lamb}
\Delta(\mu,p\,; \cdot)=1-\mu \Omega(p\,;\cdot),
\end{equation}
where
\begin{equation}\label{eq3.1}
\Omega(p;z)=\int\limits_{\T^3}
\frac{\varphi^2(s)ds}{z-w_p(s)},\qquad p\in \T^3, \quad z\in
\mathrm{C}\backslash [m(p); M(p)].
\end{equation}

\begin{lemma}\label{eigenvalue}
A number $z\in \mathrm{C}\setminus \sigma_{ess}(H_{\mu }(p)),
p\in\T^3$ is an eigenvalue of the operator $H_{\mu}(p)$ if and only
if $\Delta(\mu,p\,;z)=0.$ The associated eigenfunction $f\in
L^2(\T^3)$ is of the form
\begin{equation}\label{f}
f(q)=\frac{C\mu\varphi(q)}{z-w_p(q)},
\end{equation}
where $C\neq0$ is a normalizing constant.
\end{lemma}
\begin{proof} If
a number $z\in \bbC\setminus \sigma_{ess}(H_{\mu }(p)), p\in\T^3$ is
an eigenvalue of the operator $H_{\mu}(p)$ and $f\in L^2(\T^3)$ is
an associated eigenfunction, i.e., the equation
\begin{equation}\label{nonzeroz}
[\omega_p(q)-z]f(q)-\mu \varphi(q) \int\limits_{{\T}^3}
\varphi(t)f(t)dt=0,
\end{equation}
with
\begin{equation*}
\int\limits_{{\T}^3} \varphi(t)f(t)dt\neq 0.
\end{equation*}
has solution, then the solution $f$ of equation \eqref{nonzeroz} is
given by
\begin{equation}\label{solutionz}
f(q)=\frac{C\mu\varphi(q)}{z-w_p(q)},
\end{equation}
where $C\ne0$ is a normalizing constant. The representation
\eqref{solutionz} of the solution of equation \eqref{nonzeroz}
implies that $\Delta(\mu,p\,;z)=0$.

Conversely. Let $\Delta(\mu,p\,;z)=0$ for some $z\in \bbC\setminus
\sigma_{ess}(H_{\mu }(p)), p\in\T^3$. Then the function $f$, defined
by \eqref{solutionz}, belong to $L^2(\T^3)$ and obeys the equation
$H_{\mu}(p)f=zf.$

The analyticity of the eigenfunction $f(\cdot)$ defined by
\eqref{solutionz} follows from the analyticity of $\varphi(\cdot)$
and $w_p(\cdot)$ as well as due to the fact that the denominator
$z-w_p(\cdot)$ in \eqref{solutionz} is not vanished.

\end{proof}

\begin{proposition}\label{proposition}
For $\zeta<0$  the following equalities hold:
$$I_{n}(\zeta) =\int\limits_{0}^\delta
\frac{r^{2n}dr}{r^2-\zeta}=
\dfrac{\pi}{2}\cdot\dfrac{\zeta^n}{\sqrt{-\zeta}}+\tilde{I}_n(\zeta)\,
,\,n=0,1,2,\cdots,
$$ where $\tilde{I}_n(\zeta)$ is an
analytic function in a neighborhood of the origin {\cite{L92}}.
\end{proposition}

\begin{lemma}\label{integral}
Assume Hypothesis \ref{Hyp1}.Then for any $p\in U_{\delta}(p_{\,0})$
the integral
$$\Omega(p)=\Omega(p, M(p))=\int\limits_{\T^3} \dfrac{\varphi^2(s)ds}{M(p)-w_p(s)
}\,$$ exists and defines an  analytic function in
$U_{\delta}(p_{\,0}).$
\end{lemma}

\begin{proof}
We represent the function
$$\Omega(p,z)=\int\limits_{\T^3} \dfrac{\varphi^2(s)ds}{z-w_p(s)}\,$$ in the form
\begin{gather} \label{Omega}\Omega(p,z)=
\int\limits_{U(\mathbf{q}_0(p))}\dfrac{\varphi^2(s)ds}{z-w_p(s)}+\int\limits_{\T^3\setminus
U(q_0(p))}
\dfrac{\varphi^2(s)ds}{z-w_p(s)}=\\
=\Omega_1(p,z)+\Omega_2(p,z),\nonumber
\end{gather}
where
 $U(\mathbf{q}_0(p))$ is a neighborhood of $\mathbf{q}_0(p).$

Observe that by Hypothesis \ref{Hyp1} for any $p\in
U_{\delta}(p_{\,0})$ the function $\Omega_2(p,z)$ is analytic at
$z=M(p).$

We note that by the parametrical Morse lemma for any $p\in
U_{\delta}(p_{\,0})$ there exists a map $s=\psi(y,p)$ of the sphere
$W_{\gamma}(0)\subset \mathbb{R}^3$ with radius $ \gamma>0$ and
center at $y=0$ to a neighborhood $U(\mathbf{q}_0(p))$ of the point
$\mathbf{q}_0(p)$ that in $U(\mathbf{q}_0(p))$ the function
$w_p(\psi(y,p))$ can be represented as
$$w_p(\psi(y,p))=M(p)-y_1^2-y_2^2-y_3^2=M(p)-y^2.$$ Here the
function $\psi(y,\cdot)$ (resp. $\psi(\cdot,p)$) is analytic in
$U_{\delta}(p_{\,0})$ (resp. $W_{\gamma}(0)$) and
$\psi(0,p)=\mathbf{q}_0(p).$ Moreover, the Jacobian $J(\psi(y,p))$
of the mapping  $s=\psi(y,p) $ is analytic in $W_{\gamma}(0)$ and
positive, i.e., $J(\psi(y,p))>0$ for all $y\in W_{\gamma}(0)$ and
$p\in U_{\delta}(p_{\,0})$.

In the integral for $\Omega_1(p,z)$ changing of variables
$s=\psi(y,p) $ gives

\begin{equation}\Omega_1(p,z) = \int\limits_{W_{\gamma}(0)}
\dfrac{\varphi^2(\psi(y,p))}{y^2+z-M(p)}J(\psi(y,p))dy,
\end{equation}
where $J(\psi(y,p))$ is  the Jacobian of the mapping
$\psi(y,p).$

Passing to spherical coordinates as $y=r\nu$, we obtain \be
\Omega_1(p,z) = \int\limits_{0}^{\gamma}
\dfrac{r^{2}}{r^2+z-M(p)}\left\{\int\limits_{\Omega_3}
\varphi^2(\psi(r\nu,p))J(\psi(r\nu,p))\,d\nu\right\} dr,\ee where
$\Omega_3$ is a unit sphere in  $\R^3$ and $d\nu$ -- its element.
Inner integral can be represented as
\be\label{iii}\int\limits_{\Omega_3}
\varphi^2(\psi(r\nu,p))J(\psi(r\nu,p))\,d\nu
=\sum\limits_{n=0}^{\infty} \tau_n(p) r^{2n},\ee where $\tau_n(p),$
$n=0,1,2,\ldots$ are Pizetti coefficients.

Thus we have that
\be\label{x}\Omega_1(p,z)=\sum\limits_{n=0}^{\infty}
\tau_n(p)\int\limits_{0}^{\gamma} \dfrac{r^{2n+2}dr
}{r^2+z-M(p)},\quad
\tau_0(p)=\varphi^2(\mathbf{q}_0(p))J(\mathbf{q}_0(p))\ee Since the
function under the integral sign in \eqref{iii} is analytic in
$U_{\delta}(p_{\,0}),$ the Pizetti coefficients $\tau_n(p),$
$n=0,1,2,\ldots$ are analytic in $U_{\delta}(p_{\,0}).$ The
representation \eqref{x} yields that the following limit exists
$$\Omega_1(p)=\lim\limits_{z\to M(p)+0}\Omega_1(p,z)=\lim\limits_{z\to M(p)+0}
\sum\limits_{n=0}^{\infty} \tau_n(p)\int\limits_{0}^{\gamma}
\dfrac{r^{2n+2}dr}{r^2+z-M(p)}=\sum\limits_{n=0}^{\infty}
\frac{\gamma^{2n+1}}{2n+1}\,\tau_n(p)$$ and consequently,
\be\label{Omega_1}\Omega(p)=\lim\limits_{z\to
M(p)+0}\Omega(p,z)=\Omega_1(p)+\Omega_2(p),\ee where
$\Omega_2(p)=\Omega_2(p,M(p)).$ The analyticity of the Pizetti
coefficients $\tau_n(p),$ $n=0,1,2,\ldots$  in $U_{\delta}(p_{\,0})$
yield that the function $\Omega_1(p)$ is analytic in $p\in
U_{\delta}(p_{\,0}).$ So, $\Omega(p)$ is analytic in $p\in
U_{\delta}(p_{\,0}).$
\end{proof}

\begin{lemma}\label{lresonance}
Assume Hypothesis \ref{Hyp1} and $p\in U_{\delta}(p_{\,0}).$ Then
the following statements are equivalent:
\item[{\rm (i)}] the threshold $M(p)$ is a resonance of the operator $H_\mu(p)$ and the
associated resonance state  is of the form
\begin{equation}\label{ressolution}
f(q)=\frac{C\mu(p)\varphi(q)}{M(p)-w_p(q)},
\end{equation}
where $C\ne0$ is a normalizing constant.
\\
\item[{\rm (ii)}] $\varphi(\mathbf{q}_0(p))\neq 0$ and $\Delta(\mu,p\,;M(p))=0.$ \\
\item[{\rm (iii)}] $\varphi(\mathbf{q}_0(p))\neq 0$ and $\mu= \mu(p).$
\end{lemma}

\begin{proof}
Let the threshold $M(p)$ be a resonance of the operator $H_\mu(p)$.
According to the definition of resonance the equation
$H_{\mu}(p)f=M(p)f$ has a nontrivial solution $f\in
L^1(\T^3)\setminus L^2(\T^3)$, i.e., the equation
\begin{equation}\label{eqresonance}
[M(p)-w_p(q)]f(q)-\mu \varphi(q) \int\limits_{{\T}^3}
\varphi(t)f(t)dt=0,
\end{equation}
with
\begin{equation*}
\int\limits_{{\T}^3} \varphi(t)f(t)dt\neq 0.
\end{equation*}
has a nontrivial solution. It is easy to check that the solution $f$
of  equation \eqref{eqresonance}, i.e., the resonance state, is
given by \eqref{ressolution}. Since $w_p(\cdot)$ has a unique
non-degenerate maximum at $\mathbf{q}_0(p)\in \T^3$, in the integral
$$\Omega_1(p)=\int\limits_{W_\gamma(0)}\frac{
 \varphi^2(\psi(y,p))J(\psi(y,p))dy }{y^4}.$$
passing to spherical coordinates as $y=r\nu$ we get
\be\label{asymptotics1} \Omega_1(p)=\int\limits_0^\gamma
\left(\int\limits_{\Omega_3}
 \varphi^2(\psi(r\nu,p))J(\psi(r\nu,p)) \,d\nu\,\right)r^{-2}
 dr.\ee
Expanding the function $\varphi(\psi(r\nu,p))$ to the Taylor series
at $r=0$ we obtain \be\label{asymptotics2}
\varphi(\psi(r\nu,p))=\varphi(\mathbf{q}_0(p))+\sum\limits_{i=1}^3
\frac{\partial\varphi}{\partial
\psi^{(i)}}(\mathbf{q}_0(p))\left(\sum\limits_{j=1}^3\frac{\partial\psi^{(i)}}{\partial
y_j}(0,p)\,\nu_j\right)r+g(r,\nu)r^2,\,y_j=r\nu_j, \ee where
$g(\cdot,\nu)$ is continuous in $W_{\gamma}(0)$ and
$\nu^2_1+\nu^2_2+\nu^2_3=1.$  Since the solution $f$ of the equation
\eqref{eqresonance} belongs to $L^1(\T^3)\setminus L^2(\T^3)$ the
asymptotics \eqref{asymptotics2} yields the relation
$\varphi(\mathbf{q}_0(p))\ne0$.

Putting the expression \eqref{ressolution} for $f$ to the equation
\eqref{eqresonance} yields
\begin{equation}\label{resonance}
\varphi(q)-\mu \varphi(q) \int\limits_{{\T}^3}
\frac{\varphi^2(t)dt}{M(p)-w_p(t)}=0,
\end{equation}
which implies the equalities $\Delta_{\mu}(p,M(p))=0$, and
$\mu=\mu(p)$.

Let $\varphi(\mathbf{q}_0(p))\neq 0$ and $\mu= \mu(p)$. Then it easy
to check that $\Delta(\mu,p\,;M(p))=0$ and the function $f$, defined
by \eqref{ressolution}, belongs to $L^1(\T^3)\setminus L^2(\T^3)$
and obeys the equation $H_{\mu}(p)f=M(p)f.$

\end{proof}

\begin{lemma}\label{leigenvalue}
Assume Hypothesis \ref{Hyp1} and $p\in U_{\delta}(p_{\,0}).$ Then
the following statements are equivalent:
\item[{\rm (i)}] The threshold $z=M(p)$ is an eigenvalue of the operator $H_\mu(p)$ and the
associated eigenvector  is of the form
\begin{equation}\label{eigsolution}
f(q)=\frac{C\mu(p)\varphi(q)}{M(p)-w_p(q)},
\end{equation}
where $C\ne0$ is a normalizing constant.
\item[{\rm (ii)}] $\varphi(\mathbf{q}_0(p))= 0$  and $\Delta(\mu,p\,;M(p))=0.$ \\
\item[{\rm (iii)}] $\varphi(\mathbf{q}_0(p))= 0$ and $\mu= \mu(p).$
\end{lemma}

\begin{proof}
Let $z=M(p)$ be an eigenvalue of the operator $H_{\mu}(p)$ and $f\in
L^2(\T^3)$ is an associated eigenfunction, i.e., the equation
\begin{equation}\label{eigenequation}
[M(p)-w_p(q)]f(q)-\mu \varphi(q) \int\limits_{{\T}^3}
\varphi(t)f(t)dt=0,
\end{equation}
with
\begin{equation*}
\int\limits_{{\T}^3} \varphi(t)f(t)dt\neq 0,
\end{equation*}
has nontrivial solution. Then the associated eigenfunction $f$ is
given by \eqref{eigsolution}. In this case the relation $ f\in
L^2(\T^3)$ and asymptotics \eqref{asymptotics2} yield the equality
$\varphi(\mathbf{q}_0(p))= 0$.  The equation \ref{eigenvalue}
implies the equality $\Delta(\mu,p\,;M(p))=0$, which yields that
$\mu=\mu(p)$

 Let $\varphi(\mathbf{q}_0(p))= 0$ and $\mu=\mu(p)$. Then
$\Delta(\mu,p\,;M(p))=0$ and the function $f$, defined by
\eqref{eigsolution}, obeys the equation $H_{\mu}(p)f=M(p)f.$
\end{proof}
\begin{corollary}\label{trivial}
The equation $H_{\mu}(p)f=M(p)f$ has only trivial solution $f=0\in
L^2(\T^3)$ if and only if $\mu\ne \mu(p)$.
\end{corollary}

In the following lemma we establish an expansion for
$\Delta(\mu,p\,;z)$ in a half-neighborhood $(M(p), \allowbreak
M(p)+\delta)$ of the point $z=M(p).$

\begin{lemma}\label{representation}
Assume Hypothesis \ref{Hyp1}. Then for any $\mu>0,$ $p\in
U_{\delta}(p_{\,0})$ and sufficiently small  $z-M(p)>0$ the function
$\Delta(\mu,p; \cdot)$ can be represented as  following convergent
series \be\label{30}\Delta(\mu,p\,;z)=1-\mu
\Omega(p)+\mu\frac{\pi\tau_0(p)}{2}(z-M(p))^{1/2}-\mu\sum\limits_{n=2}^{\infty}
c_n(p)(z-M(p))^{n/2},\ee
$$ \tau_0(p)=\varphi^2(\mathbf{q}_0(p))J(\mathbf{q}_0(p)),$$
\end{lemma}

\begin{proof}
According to \eqref{x} and Proposition \ref{proposition} the
function $\Omega_1(p;z)$ can be written as
\be\label{22}\Omega_1(p,z)=
-\frac{\pi\tau_0(p)}{2}(z-M(p))^{1/2}+\sum\limits_{n=1}^{\infty}
\tilde{c}_n(p) (z-M(p))^{n+1/2}+\tilde{F}(p,z),\ee where
$\tilde{{F}}(p,z)$ is an analytic function at the point  $z=M(p)$
and
$$\tilde{c}_n(p)=\dfrac{(-1)^{n+1} \pi \tau_n(p)}{2}.$$

Consequently, the decomposition \eqref{Omega} yields for
$\Omega(p,z), z\in [M(p),\, M(p)+\delta)$ the following representation
$$\Omega(p,z)=
-\frac{\pi\tau_0(p)}{2}(z-M(p))^{1/2}+\sum\limits_{n=1}^{\infty}
\tilde {c}_n(p) (z-M(p))^{n+\frac{1}{2}S}+{F}(p,z),$$
where $F(p,z)=\tilde{{F}}(p,z)+\Omega_2(p,z)$
is analytic function at the point  $z=M(p).$

Notifying
$F(p,M(p))=\Omega(p)$ and $(z-M(p))^{1/2}>0$ for  $z>M(p)$, we
obtain
$$\Omega(p,z)=
\Omega(p)-\frac{\pi\tau_0(p)}{2}(z-M(p))^{1/2}+\sum\limits_{n=2}^{\infty}
c_n(p)(z-M(p))^{n/2}.$$  The equality
\eqref{Det.H.lamb} proves Lemma \ref{representation}.

\end{proof}

Now we prove the main results.

\begin{proof} [\bf Proof of Theorem \ref{main theorem}]
{(i)}  Let $\mu>\mu(p).$
Then Lemma \ref{representation} gives that
$$\lim\limits_{z\to M(p)+0}\Delta(\mu,p\,;z)=\Delta(\mu,p\,;M(p))=1-\frac{\mu}{\mu(p)}<0.$$
The function $ \Delta(\mu,p\,;\cdot)$ is continuous and monotonously
increasing in $z\in (M(p), +\infty)$ and
\be\label{lim}\lim\limits_{z\to+\infty}\Delta(\mu,p\,;z)=1.\ee
Whence, $ \Delta(\mu,p\,;z)=0$ for a unique $z\in (M(p), +\infty).$

Let $\Delta(\mu,p\,;z)=0$ for some $z\in (M(p), +\infty)$. Then
\begin{equation}\label{determinant}
1-\frac{\mu}{\mu(p)}=\Delta(\mu,p\,;M(p))<\Delta(\mu,p \,;z)=0
\end{equation}
which yields that $\mu>\mu(p)$. The Lemma \ref{leigenvalue} ended
the proof of the statement.

Since $z=E(\mu,p) $ is a solution of the equation
$\Delta(\mu,p\,;z)=0$ and $\Delta(\mu,\cdot;z)$ (resp.
$\Delta(\cdot,p\,;z)$) is real-analytic in $U_{\delta}(p_{\,0})$
(resp. $(\mu(p),+\infty)$), the implicit function theorem implies
that $E(\mu,\cdot)$ (resp. $E(\cdot,p)$) is real analytic in
$U_{\delta}(p_{\,0})$  (resp. $(\mu(p),+\infty)$).

Note that $p\in U_{\delta}(p_{\,0})$ the function
$\Delta(\cdot,p;z)$  monotonously decreases  in $(\mu(p),+\infty)$
and hence  the solution (eigenvalue) $ E(\mu,p)$  also monotonously
decreases  in $(\mu(p),+\infty)$.

Lemma \ref{eigenvalue} implies that if the number  $E(\mu,p) $ is an
eigenvalue of $H_\mu(p),$ $p\in U_{\delta}(p_{\,0}),$ then the
function
$$\Psi(\mu;p,\cdot,E(\mu,p))=\frac{C\mu
\varphi(\cdot)}{E(\mu,p)-w_p(\cdot)} ,$$ where $C\neq0$ is a
normalization constant, is a solution of the equation
$$H_{\mu}(p)\Psi(\mu;p,q,E(\mu,p))=
E(\mu,p)\Psi(\mu;p,q,E(\mu,p)).$$

The analyticity of  $\Psi(\mu;p,\cdot,E(\mu,p))$ follows from the
analyticity of $\varphi(\cdot)$ and $(w_p(\cdot)-E(\mu,p))^{-1}$ in
$\T^3$.

Since the functions $E(\mu,\cdot)$ (resp. $E(\cdot,p)$) and
$w(\cdot,q)$ are analytic in $U_{\delta}(p_{\,0})$ (resp.
$(\mu(p),+\infty)$) and $w_p(q)-E(\mu,p)>0$ the mapping $p\mapsto
\Psi(\mu;p,q,E(\mu,p))$ (resp. $\mu\mapsto \Psi(\mu;p,q,E(\mu,p))$)
is also analytic mapping in $U_{\delta}(p_{\,0})$ (resp.
$(\mu(p),+\infty)$).

We can prove the rest part of statements of Theorem \ref{main
theorem} applying Lemmas \ref{lresonance} and \ref{leigenvalue} by
the same way as the proof of $(i)$.
\end{proof}
\section{Acknowledgments}
The first author would like to thank the School of Mathematical
Sciences, Faculty of Science and Technology, University Kebangsaan
Malaysia for the invitation and hospitality. The work was supported
by the
Fundamental Science Foundation of Uzbekistan and by the grant no.\\
ERGS/1/2/2013/STG06/UKM/01/2.

\end{document}